\documentclass[11pt]{article}

\usepackage[margin=1in]{geometry}
\usepackage{amsmath,amssymb,amsthm,mathtools}
\newtheorem{theorem}{Theorem}
\newtheorem{claim}{Claim}
\newtheorem{lemma}{Lemma}
\newtheorem{remark}{Remark}
\usepackage{enumitem}
\usepackage{hyperref}
\usepackage[ruled,vlined,linesnumbered]{algorithm2e}
\usepackage{tikz}

\definecolor{layerone}{RGB}{86,180,233}
\definecolor{layertwo}{RGB}{230,159,0}
\definecolor{layerthree}{RGB}{0,158,115}
\definecolor{layerfour}{RGB}{204,121,167}
\definecolor{layerfive}{RGB}{213,94,0}

\newcommand{\E}{\mathbb{E}}
\newcommand{\Prb}{\mathbb{P}}
\newcommand{\occ}{\operatorname{occ}}
\newcommand{\poly}{\operatorname{poly}}
\newcommand{\wt}{\widetilde}
\newcommand{\del}{\delta}
\newcommand{\hx}{\hat x}
\newcommand{\defn}[1]{\emph{#1}}
\usepackage{todonotes}

\title{Linear Probing with Non-Greedy Insertions}
\newif\ifanonymous
\anonymousfalse 
\ifanonymous
\author{Anonymous submission}
\hypersetup{pdfauthor={}}
\else
\author{Andrew Krapivin\thanks{Carnegie Mellon University.}
\and William Kuszmaul\thanks{Carnegie Mellon University.}
\and Yixuan Wang\thanks{Carnegie Mellon University.}}
\hypersetup{pdfauthor={Andrew Krapivin, William Kuszmaul, Yixuan Wang}}
\fi
\date{}

\begin{document}
\maketitle

\begin{abstract}
Linear probing hash tables classically use a \emph{greedy} insertion strategy, placing a key $u$ in the first available position out of $h(u), h(u) + 1, h(u) + 2, \ldots$. If the hash table is filled to $1 - 1/x$ full, this results in $\Theta(x^{2})$ worst-case expected insertion time.

In this paper, we introduce interlinear probing, a simple \emph{non-greedy} insertion strategy that does better without requiring elements to be reordered within the table over time. Given $x$ in advance, the algorithm brings the worst-case expected insertion time (and therefore also the worst-case expected positive query time) down to $O(x \log x)$. We also extend the algorithm to support \emph{negative} queries in worst-case expected time $O(x (\log x)^2)$. Moreover, our construction achieves $O(x)$ amortized insertion time, matching that of standard linear probing.

Finally, we prove a lower bound showing that any stable insertion strategy for linear probing must incur $\Omega(x \sqrt{\log x})$ worst-case expected time. Combined, our results establish that the optimal worst-case expected insertion time among stable insertion strategies is $x (\log x)^{\Theta(1)}$.
\end{abstract}

\section{Introduction}

The \emph{linear-probing hash table} is one of the oldest and most widely used data structures in computer science. In its simplest form, the hash table works as follows. Elements are stored in an array of some size $n$, and insertions place an element $u$ in the first available position from the sequence $h(u), h(u) + 1, h(u) + 2, \ldots \bmod n$, where $h(u)$ is a random hash. Queries can then locate an element by examining the same probe sequence $h(u), h(u) + 1, h(u) + 2, \ldots$ until they encounter the position containing $u$.

What makes linear probing appealing is its \emph{data locality}---each operation accesses a single small region of contiguous memory \cite{richter2015seven}. The main drawback of linear probing, on the other hand, is \emph{clustering} \cite{Knuth63,benderlinearprobing}. As the hash table fills up, elements have a tendency to cluster together into long runs. The result is that, when the hash table is filled to a load factor of $1 - 1/x$, the expected insertion and query times both grow at a rate of $\Theta(x^{2})$, rather than the rate of $\Theta(x)$ that one might intuitively expect \cite{Knuth63,benderlinearprobing}.

To combat clustering, past works have proposed several techniques for rearranging elements in the hash table over time in order to speed up both insertions and queries. Notably, by keeping the elements in each run in sorted order by hash \cite{AmbleKn74,CelisLaMu85,celis1986robin,benderlinearprobing}, one can reduce the expected query time to $O(x)$; and by strategically rebuilding the hash table at regular intervals, and planting \emph{tombstones} throughout the table during each rebuild, one can bring the expected insertion time down to $O(x)$ as well \cite{benderlinearprobing}.

These techniques rely fundamentally on the ability to move elements around over time. But in many settings \cite{richter2015seven,sandersstability,iceberg,icebergimplementation,maier2022scalable,mccoy2026warpspeed}, it is desirable for a hash table to be stable, meaning that elements do not move around after they are inserted. This raises the following basic question: Is there a way to build a \emph{stable} linear-probing insertion strategy that achieves worst-case $O(x)$ expected-time insertions?

\paragraph{The challenge: \emph{average} insertion time must be $\Omega(x)$. } Part of what makes the above question challenging is the following basic fact: no matter how the $(1 - 1/x)n$ elements are placed in a linear-probing hash table, the \emph{average} distance from each element $u$'s hash $h(u)$ to the position where $u$ resides must be $\Omega(x)$ \cite{peterson57,Knuth63}.\footnote{Concretely, Peterson \cite{peterson57} observed that the sum of these distances is minimized by using the standard ordering; Knuth \cite{Knuth63} then showed that the standard ordering results in average distance $\Theta(x)$.}  It follows that any stable insertion strategy must incur $\Omega(x)$ time, on average, across \emph{all} insertions.

Thus, if we want to achieve worst-case expected time close to $O(x)$, then we must incur expected time $\approx x$ not just for the final insertions, but for \emph{almost all} insertions. Moreover, the early insertions must somehow make \emph{productive use} of their insertion times in a way that allows later insertions to be cheaper. Whether or not such an insertion strategy exists is the main question of this paper.

\paragraph{This paper: upper and lower bounds of $x (\log x)^{\Theta(1)}$. }
For our upper bounds, we assume $x\le n^{0.49}$. Our first result is a simple implementation of stable linear probing that, given $x$ in advance, supports all of the first $(1 - 1/x)n$ insertions (and therefore also positive queries) in worst-case expected $O(x \log x)$ time. The data structure can be viewed as a non-greedy variation of linear probing. As in standard linear probing, queries/insertions still measure their time cost by the number of probes $h(u), h(u) + 1, h(u) + 2, \ldots$ that are performed. But, unlike in traditional linear probing, insertions are non-greedy: an insertion may choose to \emph{skip} over a free slot, in order to leave that slot for future insertions. This allows early insertions to perform extra work (strategically skipping over certain free slots) in order to make later insertions asymptotically cheaper.

Given that the proposed insertion strategy is non-greedy, it is natural to ask whether \emph{negative} queries (queries for elements that are not in the hash table) can still be supported efficiently. In a standard linear-probing hash table, a negative query to an element $u$ can stop as soon as it finds a free slot in the sequence $h(u), h(u) + 1, h(u) + 2, \ldots$. But, with non-greedy insertions, this protocol is no longer valid.

Our second result is an extension of the above data structure to support negative queries in worst-case expected time $O(x (\log x)^2)$. Specifically, we show that if elements are placed into the hash table in the right way, then negative queries can follow a new stopping rule, stopping not when they find a free slot, but instead when a certain more subtle condition is met. This allows for negative queries to complete efficiently even though the insertion strategy is non-greedy.

Our third result further extends the data structure so that the amortized expected insertion cost (and thus the expected cost to perform a positive query on a random element) is $O(x)$. Putting our upper bound results together, we have a non-greedy linear probing scheme, which we call \emph{interlinear probing}, with worst-case expected insertion and positive query times of $O(x\log x)$ and a worst-case expected negative query time of $O(x\log^2 x)$, all while maintaining an amortized expected insertion time of $O(x)$.

\begin{theorem}
\label{thm:main-amortized}
There is an implementation of linear probing without reordering, namely interlinear probing, that has the following properties. If $x\le n^{0.49}$ is known in advance, then each of the first $(1-1/x)n$ insertions completes in worst-case expected time $O(x\log x)$, each positive query completes in worst-case expected time $O(x\log x)$, and each negative query completes in worst-case expected time $O(x\log^2 x)$. Moreover, the amortized expected insertion cost and the average expected positive-query cost are $O(x)$.
\end{theorem}

Finally, the most technical result of the paper is a lower bound showing that an insertion bound of $O(x)$ is \emph{not} possible. Concretely, we prove that any stable insertion strategy for linear probing must incur $\Omega(x \sqrt{\log x})$ worst-case expected time for some insertion. Combined with our upper bounds, this establishes that the optimal worst-case expected insertion time among stable insertion strategies is $x (\log x)^{\Theta(1)}$.

An interesting feature of our lower bound is that it applies not just to stable insertion strategies, but to any insertion strategy where the cost of moving an element from a slot $i$ to a slot $j$ is measured as the cyclic distance between \(i\) and \(j\) (here we treat element $u$ as initially belonging to slot $h(u)$). Any such insertion strategy must incur $\Omega(x \sqrt{\log x})$ worst-case expected insertion time.


\paragraph{Other related work. } One of the major questions that theoreticians care about in the study of hash table design is the question of whether one can build high-performance hash tables that are stable \cite{yao1985uniform,farach2024optimal,farach2026greedy,demaine2005dynamic,iceberg}. Stability is also important from a practical standpoint, as it allows one to avoid the overheads of moving elements around \cite{icebergimplementation,richter2015seven}, it enables simpler and faster concurrency \cite{mccoy2026warpspeed,icebergimplementation,maier2022scalable}, and it allows users to maintain pointers to elements without worrying about those pointers becoming invalid \cite{maier2022scalable,sandersstability,cplusplus1,cplusplus2,abseil,F14}. Stability has been studied extensively both in hash tables (like those in this paper) that are insertion-only \cite{yao1985uniform,farach2024optimal,farach2026greedy} and in hash tables that support both insertions and deletions \cite{demaine2005dynamic,iceberg}. The basic question that we consider, however, of whether stable linear probing insertions can achieve $o(x^2)$ worst-case expected time, does not appear to have been considered before.

In concurrent work \cite{zamirlocality}, Zamir studies the optimal data locality of insertions in any stable open-addressed hash table. In addition to several other results, \cite{zamirlocality} presents a hash table that achieves a related guarantee to the one in this paper: the construction can be viewed as a non-greedy \emph{bi-directional} implementation of linear probing (meaning queries must search in both directions) with worst-case expected time $O(x \log^3 x)$. The paper \cite{zamirlocality} also poses as an open question whether a solution with worst-case expected time $O(x)$ may be possible. Our lower bound offers a partial resolution to this question, showing that such a solution is not possible in any hash table that selects its hashes $h(u)$ uniformly at random. It seems likely that, with additional effort, the lower bound can be extended to apply also to hash tables that select their hashes non-uniformly (such hash tables are included in Zamir's analyses), in which case one would obtain a full negative resolution to Zamir's question.

 \section{Preliminaries}

Throughout the paper, we consider a circular hash table with $n$ slots, and we assume access to a fully random hash function $h: U \rightarrow [n]$, where $U$ is the universe of keys. We require insertions to use linear probing, meaning that the insertion of a key $u$ examines positions $h(u), h(u) + 1, h(u) + 2, \ldots \bmod n$ and eventually places the key $u$ in one of those positions. Unlike in traditional linear probing, however, we do not require the insertion to stop at the first free slot it finds. In other words, insertions are allowed to be non-greedy. We consider the scenario in which the hash table receives $(1 - 1/x)n$ insertions, with $x$ known in advance. We will assume for simplicity that $n$ is a power of $2$; with careful rounding, all of the same arguments also apply for arbitrary $n$.

\section{The Basic Insertion Algorithm}
In this section, we present the basic insertion algorithm that interlinear probing uses, which achieves the worst-case expected insertion and positive query times of Theorem~\ref{thm:main-amortized}, and which we will later extend to optimize negative queries and amortized insertions. Throughout the paper, we will make use of two global constants $C_0, \lambda$, both of which are taken to be sufficiently large positive constants.

Identify the slots with $\{0,1,\ldots,n-1\}$. For $1\le i\le \log_2 n$, let layer $L_i$ contain the positions congruent to $2^{i-1}$ modulo $2^i$ (see Figure~\ref{fig:interleaved-layers}). Thus, the layers are disjoint, layer $L_i$ consists of one slot every $s_i=2^i$ positions, and it has size $N_i=n/2^i$. The single remaining position is placed in a terminal residual layer, which is never reached before the target load under our assumptions. Every insertion will be \emph{routed} to a layer, meaning that the insertion greedily places itself in the first free slot that it sees in that layer, skipping over any free slots that it sees in other layers.

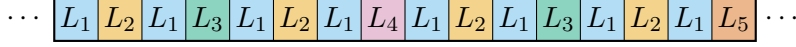
\begin{figure}[t]
\centering
\begin{tikzpicture}[x=0.58cm,y=0.58cm]
    \foreach \j in {1,3,...,15} {
        \fill[layerone!45] ({\j-1},0) rectangle (\j,1);
        \node at ({\j-0.5},0.5) {$L_1$};
    }
    \foreach \j in {2,6,...,14} {
        \fill[layertwo!45] ({\j-1},0) rectangle (\j,1);
        \node at ({\j-0.5},0.5) {$L_2$};
    }
    \foreach \j in {4,12} {
        \fill[layerthree!45] ({\j-1},0) rectangle (\j,1);
        \node at ({\j-0.5},0.5) {$L_3$};
    }
    \fill[layerfour!45] (7,0) rectangle (8,1);
    \node at (7.5,0.5) {$L_4$};
    \fill[layerfive!45] (15,0) rectangle (16,1);
    \node at (15.5,0.5) {$L_5$};

    \draw[thick] (0,0) rectangle (16,1);
    \foreach \j in {1,...,15} {
        \draw (\j,0) -- (\j,1);
    }
    \node at (-0.65,0.5) {$\cdots$};
    \node at (16.65,0.5) {$\cdots$};
\end{tikzpicture}
\caption{A decomposition of the slots into layers. Each slot is labeled and colored according to which layer it belongs to.}
\label{fig:interleaved-layers}
\end{figure}

Define
\[
\widehat x := C_0 x\log(2x),
\]
where $C_0$ is a sufficiently large absolute constant. The algorithm proceeds in phases, where in phase $i$ all insertions are routed to one of the layers $L_i$ or $L_{i + 1}$ (with high probability). During phase $i$, $L_i$ is referred to as the \emph{active layer} and layer $L_{i + 1}$ is referred to as the \emph{overflow layer}.

Let $\occ(L_i)$ denote the number of keys currently stored in layer $L_i$, and let
\[
\delta_i:=\frac{N_i-\operatorname{occ}(L_i)}{N_i}
\]
denote the current empty fraction of $L_i$. During phase $i$, the insertion algorithm routes a key to $L_i$ with probability\footnote{The more natural-looking choice $p_i(\delta_i)=\min\{1,\lambda\widehat x\delta_i^2/s_i\}$ would suffice for positive queries. We make this slightly less natural choice to facilitate efficient negative queries without needing to redo any analysis.}
\begin{equation}
p_i(\delta_i)
:=
\min\left\{
1,
\lambda\left(\frac{\widehat x}{s_i}\right)^{1/2}\delta_i^{3/2}
\right\},
\label{eq:p_i}
\end{equation}
where $\lambda$ is a sufficiently large absolute constant. With the remaining probability, the key is routed to $L_{i+1}$. We end phase $i$ and begin phase $i + 1$ when
\[
\delta_i\le \delta_i^\star,
\qquad
\delta_i^\star:=\frac{s_i}{\widehat x}.
\]

The full insertion algorithm is given in Algorithm \ref{alg:non-greedy-insert}.

\begin{algorithm}[H]
\caption{\(\textsc{InterleavedLayeredInsert}\)}
\label{alg:non-greedy-insert}
\KwIn{A key \(a\), current phase \(i\).}
\While{$\occ(L_i) \ge (1-\delta_i^\star) N_i$}{
    Set \(i\leftarrow i+1\)\;
}
Set
\[
\delta_i\leftarrow \frac{N_i-\operatorname{occ}(L_i)}{N_i},
\qquad
p_i\leftarrow
\min\left\{1,\lambda\left(\frac{\widehat x}{s_i}\right)^{1/2}\delta_i^{3/2}\right\}.
\]\\
With probability \(p_i\), insert \(a\) into \(L_i\) by ordinary linear probing restricted to \(L_i\)\;
Otherwise, insert \(a\) into \(L_{i+1}\) by ordinary linear probing restricted to \(L_{i+1}\)\;
\end{algorithm}

The above algorithm could in principle fail for one of two reasons. The first is that the condition $\occ(L_i) \ge (1-\delta_i^\star) N_i$ is met for all phases $i$ (so the while loop in the insertion protocol loops indefinitely); we will see that, so long as the global load factor is at most $1 - 1/x$, then this failure mode does not occur. The second reason the algorithm could fail is if the layer that an insertion is assigned to is already full. We will show that this happens with probability at most $1 / n^2$---in the $1 / n^2$ failure case, one can afford to complete the insertion in $\Theta(n)$ time using the trivial greedy insertion algorithm.

\section{The Analysis}

In this section, we prove the following theorem:
\begin{theorem}
Assuming $x = o(n / \log^2 n)$, each of the first $(1 - 1/x)n$ insertions completes in worst-case expected time $O(x \log x)$.
\label{thm:main}
\end{theorem}

\paragraph{Verifying that the algorithm supports load factor $1-1/x$. } Before we can analyze expected insertion time, we must first verify correctness. By construction, the algorithm stops supporting insertions once it gets to a phase $i$ where $\delta_i^\star > 1$ (at which point phase $i$ and all subsequent phases are empty by construction). We therefore must verify that the algorithm reaches global load factor $1-1/x$ before it gets to any such phase.

The total number of free slots when the algorithm completes (i.e., when it has finished all phases) is at most
\begin{align*}
1+\sum_{i=1}^{\log_2 n} \min(|L_i|, \delta_i^\star N_i) &= 1+\sum_{i=1}^{\log_2 n} \min(N_i, (s_i/\widehat x) N_i) \\
&= 1+\sum_{i=1}^{\log_2 n} \min(n / 2^i, n /\widehat x) \tag{since $N_i = n / s_i = n / 2^i$},
\end{align*}
where the extra $1$ accounts for the terminal residual layer.
Recalling that $\widehat x = C_0 x \log(2x)$ for a sufficiently large constant $C_0$, this sum is at most
\begin{align*}
    1+\sum_{i \ge 4 \log_2 x} n / 2^i +  4 (\log_2 x) \cdot n / \widehat x & \le n / x.
\end{align*}
Therefore, the algorithm reaches global load factor $1-1/x$ before all phases complete. This completes the proof of correctness.

\paragraph{Intuition for how to analyze cost. } Next we turn to bounding the expected cost of an insertion. Before presenting the formal analysis, we first sketch the intuition behind why the expected cost of an insertion is $O(x \log x)$.

Consider an insertion late in phase $i$, when the fraction $\delta_i$ of empty slots in $L_i$ is small. If an insertion is routed to $L_i$, it will incur $\Theta(\delta_i^{-2})$ expected probes to $L_i$; since slots in $L_i$ are $s_i = 2^i$ apart, this results in a total expected cost of $\Theta(s_i/\delta_i^2)$.

This is where the routing probability $p_i$ comes in. When $\delta_i$ is small, we have $p_i = \lambda (\widehat x/s_i)^{1/2}\delta_i^{3/2}$. Since $\delta_i\ge\delta_i^\star=s_i/\widehat x$ during phase $i$, the expected cost from the case where the insertion is routed to $L_i$ is therefore
$$\Theta(s_i/\delta_i^2) \cdot p_i = \Theta\left(\sqrt{\frac{\widehat x s_i}{\delta_i}}\right) = O(\widehat x) = O(x \log x).$$

The subtle part of the analysis (and indeed the algorithm design) is to ensure that the insertion is also cheap if it gets routed to $L_{i + 1}$. We will see that, if the constant $\lambda$ is large enough, then the array $L_{i + 1}$ will remain relatively empty (less than half full) with high probability for the entire duration of phase $i$. Thus insertions into $L_{i + 1}$ will be cheap, costing only $O(s_{i + 1}) = O(2^i)$ in expectation. With a bit of algebra, one can check that $2^i \le O(x \log x)$, meaning that the expected cost per insertion never exceeds $O(x \log x)$.

With this intuition in mind, we now proceed to the formal analysis.


\paragraph{Bounding the load of $L_{i + 1}$ during phase $i$.} We now show that, during phase $i$, the total number of keys routed to $L_{i+1}$ is less than $|L_{i+1}|/2$ with high probability.

Fix a phase $i$, and write $N=N_i$ and $s=s_i$. Let $O_i$ denote the number of keys routed to $L_{i+1}$ during phase $i$, and recall that the phase ends when $L_i$ reaches empty fraction
\[
\delta_i^\star=\frac{s}{\widehat x}.
\]

\begin{claim}[Overflow-count bound]
Assuming $\lambda$ is a sufficiently large constant, for every phase $i$,
\[
\Pr\left[O_i\ge \frac{N_i}{4}\right]
\le
\exp\left(-\Omega\left(\frac{n}{\widehat x}\right)\right).
\]
Equivalently, with this probability bound, phase $i$ ends before $L_{i+1}$ receives $|L_{i+1}|/2$ insertions.
\label{claim:overflow-count}
\end{claim}

The basic idea of the proof is that during phase $i$, we note that $p_i$, the probability of inserting into $L_i$, is only updated when there are new inserts into $L_i$. Thus, between each pair of consecutive insertions into layer $L_i$, the number of insertions into layer $L_{i+1}$ is a geometric random variable. Summing these geometric variables over the phase gives the total number of overflows, which can be bounded via a Chernoff bound.

\begin{proof}
Let $r$ be the number of empty slots currently remaining in $L_i$. While $r>\delta_i^\star N$, the probability that the next key is routed to $L_i$ is
\[
p_r
=
\min\left\{
1,
\lambda\left(\frac{\widehat x}{s}\right)^{1/2}
\left(\frac{r}{N}\right)^{3/2}
\right\}.
\]
Let $G_r$ be the number of overflows to $L_{i+1}$ before the next successful insertion into $L_i$, when $r$ empty slots remain. Then $G_r$ is a geometric random variable counting failures before the first success, with success probability $p_r$. Thus
\[
\mathbb E[G_r]
=
\frac{1-p_r}{p_r}
\le
\frac1{p_r}.
\]
The variables $G_r$ are independent because the routing coins are independent and $p_r$ depends only on $r$.

Only terms with $p_r<1$ can have $G_r > 0$. For these terms,
\[
p_r
=
\lambda\left(\frac{\widehat x}{s}\right)^{1/2}
\left(\frac{r}{N}\right)^{3/2}.
\]
Therefore
\begin{align*}
\sum_{r\ge \delta_i^\star N}\frac1{p_r}
&\le
\sum_{r\ge \delta_i^\star N}
\frac1\lambda\left(\frac{s}{\widehat x}\right)^{1/2}
\left(\frac{N}{r}\right)^{3/2}\\
&=
\frac{N^{3/2}}{\lambda}
\left(\frac{s}{\widehat x}\right)^{1/2}
\sum_{r\ge \delta_i^\star N}r^{-3/2}\\
&\le
\frac{3N}{\lambda}
\left(\frac{s}{\widehat x\delta_i^\star}\right)^{1/2}.
\end{align*}
The last inequality uses
$\sum_{r\ge m}r^{-3/2}\le 3m^{-1/2}$ with
$m=\lceil\delta_i^\star N\rceil$.\footnote{Indeed,
$\sum_{r\ge m}r^{-3/2}\le m^{-3/2}+\int_m^\infty t^{-3/2}\,dt
=m^{-3/2}+2m^{-1/2}\le3m^{-1/2}$.}
Since
\[
\delta_i^\star=\frac{s}{\widehat x},
\]
this becomes
\[
\sum_{r\ge \delta_i^\star N}\frac1{p_r}
\le
\frac{3N}{\lambda}.
\]
Choosing $\lambda$ large enough gives
\[
\mathbb E[O_i]\le \frac{N}{32}.
\]

We now prove concentration. The smallest relevant success probability occurs at the end of the phase, and is
\[
p_{\min}
=
\lambda\left(\frac{\widehat x}{s}\right)^{1/2}
(\delta_i^\star)^{3/2}
=
\frac{\lambda s}{\widehat x}.
\]
If this quantity is $1$ or larger, then the phase has no overflows, and the claim is trivial ($L_{i + 1}$ is empty at the end of the phase). Thus we can assume for the rest of the proof that $p_{\min} < 1$.
For a geometric random variable $G$ counting failures before success with success probability $p$, one has, for $0\le \theta\le p/4$,
\[
\mathbb E[e^{\theta G}]
\le
\exp\left(\frac{2\theta}{p}\right).
\]
Choose $\theta=p_{\min}/8$. Then $\theta\le p_r/4$ for every relevant $r$, and hence
\[
\mathbb E[e^{\theta O_i}]
\le
\exp\left(2\theta\sum_r\frac1{p_r}\right)
\le
\exp\left(\frac{6\theta N}{\lambda}\right).
\]
By Markov's inequality,
\[
\Pr\left[O_i\ge \frac{N}{4}\right]
\le
\exp\left(
-\frac{\theta N}{4}
+
\frac{6\theta N}{\lambda}
\right).
\]
Taking $\lambda$ sufficiently large, the exponent is at most $-\Omega(\theta N)$.

It remains to lower-bound $\theta N$. Recalling that $\theta = p_{\min}/8$, and
\[
p_{\min}=\frac{\lambda s}{\widehat x},
\]
we have
\[
\theta N
=
\Omega\left(\frac{sN}{\widehat x}\right)
=
\Omega\left(\frac{n}{\widehat x}\right),
\]
because $sN=\Theta(n)$. Hence
\[
\Pr\left[O_i\ge \frac{N}{4}\right]
\le
\exp\left(-\Omega\left(\frac{n}{\widehat x}\right)\right),
\]
as claimed.
\end{proof}

Note that the failure probability $\exp(- \Omega(n / \widehat x))$ from the preceding claim is at most $1/n^2$, since $\widehat x = O(x \log x)$ and $x = o(n / \log^2 n)$ by assumption. We will refer to this failure event as a \emph{critical failure}. With probability at least $1 - O(1/n)$, no critical failures ever occur. In the $O(1/n)$ failure case, one can afford to analyze every remaining insertion using the trivial $O(n)$ bound on cost. Therefore, when analyzing expected cost, we can assume that no critical failures occur.

\paragraph{Bounding the expected cost of an insertion. }
Having shown that critical failures are rare, we now bound the expected cost of an insertion. Let $i$ be the current phase. By construction, we only perform phases $i$ for which $\delta_i^\star \le 1$.
Since
$$\delta_i^\star=\frac{s_i}{\widehat x},$$
we have that
\begin{equation}
\label{eq:phase-bound}
s_i \le \widehat x.
\end{equation}

Now, consider the expected cost of the insertion. With probability $p_i$, the insertion is routed to $L_i$. An insertion into $L_i$ will examine $\Theta(\delta_i^{-2} + 1)$ expected slots in $L_i$, making for a total expected cost of $\Theta(s_i/\delta_i^2)$. The expected cost from the case where the insertion is routed to $L_i$ is therefore
\begin{align*}
\Theta(s_i/\delta_i^2) \cdot p_i
\le O\left(\sqrt{\frac{\widehat x s_i}{\delta_i}}\right)
\le O(\widehat x)
= O(x \log x),
\end{align*}
where the second inequality uses $\delta_i\ge\delta_i^\star=s_i/\widehat x$.

On the other hand, barring any critical failures, the overflow layer $L_{i + 1}$ is at most half full. The expected cost from inserting into $L_{i + 1}$ is therefore at most $O(s_{i + 1}) = O(\widehat x) = O(x \log x)$ (where the first equality follows from \eqref{eq:phase-bound}).

Thus, the overall expected cost of the insertion is at most $O(x \log x)$, which completes the proof of Theorem \ref{thm:main}.

\section{Handling Negative Queries}

Next we modify interlinear probing to support efficient \emph{negative} queries, i.e., queries that verify a given key $u$ is \emph{not} in the table. To do this, we make two modifications to our insertion algorithm:

\paragraph{Modification 1: Breaking each layer $L_i$ into two sub-layers $L_i^1$ and $L_i^2$. } For each layer $L_i$, partition it into two sub-layers $L_i^1$ and $L_i^2$, where the slots in $L_i$ alternate between $L_i^1$ and $L_i^2$. Let $g$ be a hash function mapping elements $u$ to $\{1, 2\}$.

During phase $i - 1$, and during phase $i$ up until the probability of placing an element into $L_i$ falls below $1/2$, we perform insertions into $L_i$ as follows: To place an element $u$ into $L_i$, we place it into the first free slot that we find in $L_i^{g(u)}$, ignoring any free slots in $L_i^{2-g(u)}$. We call such an insertion an \defn{early insertion into $L_i$}.

Then, once we are in phase $i$ and the probability of placing an element into $L_i$ is below $1/2$, we perform insertions into $L_i$ in a way that swaps the roles of $L_i^1$ and $L_i^2$: To place an element $u$ into $L_i$, we place it into the first free slot that we find in $L_i^{2-g(u)}$, ignoring any free slots in $L_i^{g(u)}$. We call such an insertion a \defn{late insertion into $L_i$}.

\begin{remark}\label{rem:balanced-sublayers}
This modification introduces the following minor technical complication in the analysis: In principle, due to dumb luck, the number of elements in $L_i$ that get placed into $L_i^1$ and $L_i^2$ could be very unbalanced. By a Chernoff bound, the number of elements placed in each of $L_i^1$ and $L_i^2$ will differ by at most $O(\sqrt{n \log n})$ with high probability in $n$. So long as $\sqrt{n \log n} = o(n / \widehat x)$ (recall that we ultimately fill $L_i$ to have $\Theta(n / \widehat x)$ free slots remaining), this distinction will not matter anywhere in the analysis. To ensure that this condition holds, we will add the condition that $x \le n^{0.49}$ to our analysis in this section.
\end{remark}

\paragraph{Modification 2: Deciding which elements to insert in $L_i$ during phase $i$. } Finally, during phase $i$, when deciding whether to place an element $u$ into $L_i$ or $L_{i + 1}$, we will now use the following rule. Let $f$ be a hash function mapping each element to a random real number in $[0, 1]$. We take $f$, $g$, and $h$ to be mutually independent. Let $p$ be the probability (given by \eqref{eq:p_i}) of placing the element $u$ into $L_i$. If $p \ge 1/2$, then we make the decision with a random probability-$p$ coin flip. But, if $p < 1/2$, then we place the element $u$ in $L_i$ if and only if $f(u) \in [p, 2p]$.

With this second modification in place, we have completed the description of the modified insertion algorithm.

\paragraph{Implementing queries. } To query an element $u$, we must scan the positions $h(u), h(u) + 1, h(u) + 2, \ldots$ until either we find the element (a positive query) or until we can conclude it is not present (a negative query). We will now describe how a negative query can conclude that an element $u$ is not present (in a manner much more efficient than just looking for a free slot in each $L_i^j$). We will assume in this and the next section that the condition in Remark~\ref{rem:balanced-sublayers} holds and that insertions never enter the $1/n^2$ failure case in which they perform the trivial greedy insertion algorithm for the purpose of implementing negative queries.\footnote{We can keep one metadata bit to track this failure event and scan all $n$ slots for a query if the metadata bit is set. This would only add an $O(1)$ expected cost to queries.}

Let $I$ be the set of layer-numbers $i$ such that at least one element was inserted into $L_i$ by the algorithm (so we got to phase $\ge i - 1$). To conclude that $u$ is not present, the query must certify that $u$ is not in any $L_i^j$ for any $i \in I$ and $j \in \{1, 2\}$.

If $u$ were to be in $L_i^{g(u)}$, then it would have to have been placed there by an early insertion. All elements $v$ occupying positions of $L_i^{g(u)}$ between $h(u)$ (inclusive) and the position containing $u$ would therefore have to satisfy $g(v) = g(u)$. Thus, a negative query can conclude that $u$ is not present in $L_i^{g(u)}$ as soon as it encounters (in $L_i^{g(u)}$) either a free slot or an element $v$ such that $g(v) \neq g(u)$.

On the other hand, if $u$ were to be in $L_i^{2-g(u)}$, then it would have to have been placed there by a late insertion. Moreover, $u$ would have needed to be placed into $L_i$ when the probability $p$ of placing an element into $L_i$ satisfied $f(u) \in [p, 2p]$. One consequence of this is that any element $v$ in $L_i^{2-g(u)}$ that satisfies $g(v) = g(u)$ and that is inserted into $L_i$ \emph{before} $u$ would need to satisfy $f(v) \ge p$, and therefore $f(v) \ge f(u) / 2$. Therefore, as soon as the query encounters (in $L_i^{2-g(u)}$) either a free slot or an element $v$ such that $g(v) = g(u)$ and $f(v) < f(u) / 2$, the query can conclude that $u$ is not present in $L_i^{2-g(u)}$.

Once the query has certified that $u$ is not present in either $L_i^{g(u)}$ or $L_i^{2-g(u)}$ for any $i \in I$, it can complete as a negative query.

\paragraph{Bounding expected operation cost.}
The bound on expected insertion cost (and therefore on expected positive-query cost) follows from exactly the same analysis as in Section~4.

Therefore, we focus our analysis on bounding the expected cost of a negative query. The main step here is to bound how far a query must probe in order to conclude that $u$ is not present in $L_i^j$ for a given $i \in I$ and $j \in \{1, 2\}$.

\begin{lemma}
  Let $d_i$ be the distance from $h(u)$ to the position at which the query certifies that $u$ is not in $L_i$. Then
  \[\mathbb E[d_i] = O(\widehat x).\]
\end{lemma}

\begin{proof}
  Let $d_i^{g(u)}$ and $d_i^{2-g(u)}$ be the distances at which the query certifies that $u$ is not in $L_i^{g(u)}$ and $L_i^{2-g(u)}$, respectively. We have that
  \[\mathbb E\left[d_i\right] = \mathbb E\left[\max\{d_i^{g(u)},d_i^{2-g(u)}\}\right] \leq \mathbb E\left[d_i^{g(u)}\right] + \mathbb E\left[d_i^{2-g(u)}\right].\]

  First, we bound $\mathbb E[d_i^{g(u)}]$. Let $t$ be the point in time immediately after the final early insertion into $L_i^{g(u)}$. Any position in $L_i^{g(u)}$ that is free at time $t$ is guaranteed to either still be free later on or to contain an element $v$ such that $g(v) \neq g(u)$. Therefore, $\mathbb E[d_i^{g(u)}]$ is at most the expected distance an insertion into $L_i^{g(u)}$ would travel at time $t$.

  How full is $L_i^{g(u)}$ at time $t$? With high probability in $n$, $L_i^{g(u)}$ is at most half full during phase $i - 1$ (for $i > 1$). Then, during phase $i$, time $t$ occurs when the probability of placing an element into $L_i$ falls below $1/2$. This happens when
  \[
  \lambda\left(\frac{\widehat x}{s_i}\right)^{1/2}
  \delta_i^{3/2}\approx \frac12,
  \qquad\text{or equivalently}\qquad
  \delta_i=\Theta\left(\left(\frac{s_i}{\widehat x}\right)^{1/3}\right).
  \]
  If the execution ends before this crossing occurs, then its final routing
  probability is still at least $1/2$, which gives the same lower bound on
  $\delta_i$.
  Thus at time $t$ (and with high probability in $n$), $L_i^{g(u)}$ has load at most
  \[
  1-\Omega\left(\left(\frac{s_i}{\widehat x}\right)^{1/3}\right),
  \]
  so an insertion must inspect $O((\widehat x/s_i)^{2/3})$ positions in expectation to find a free slot. Multiplying by $s_i = 2^i$ (to account for the slots in $L_i^{g(u)}$ being $\Theta(s_i)$ apart) gives
  \[
  \mathbb E[d_i^{g(u)}]
  =
  O\left(\left(\frac{\widehat x}{s_i}\right)^{2/3}s_i\right)
  =
  O(\widehat x),
  \]
  where the final inequality uses $s_i=O(\widehat x)$ for every layer that receives an insertion.

  Next, we bound $\mathbb E[d_i^{2-g(u)}]$. Redefine $t$ to be the time at which the probability of inserting into $L_i$ goes below $f(u)/4$ (if this time does not exist, let it be the first moment after phase $i$ ends). Any position in $L_i^{2-g(u)}$ that is free at time $t$ is guaranteed to either be free later or contain an element $v$ such that $g(v)=g(u)$ and $f(v)<f(u)/2$. Thus, $\mathbb E[d_i^{2-g(u)}]$ is at most the expected cost of an insertion into $L_i^{2-g(u)}$ at time $t$.

  Time $t$ occurs during phase $i$ when (or before) the probability of placing an element into $L_i$ falls below $f(u) / 4$. This happens when
    \[
    \lambda\left(\frac{\widehat x}{s_i}\right)^{1/2}
    \delta_i^{3/2}\approx \frac{f(u)}4,
    \qquad\text{or equivalently}\qquad
    \delta_i=\Theta\left(
    f(u)^{2/3}\left(\frac{s_i}{\widehat x}\right)^{1/3}
    \right).
    \]
    Even if this point never happens, the lower bound on $\delta_i$ still holds. Thus, writing $z=f(u)$,
    \[
    \mathbb E\left[d_i^{2-g(u)}\mid f(u)=z\right]
    =
    O\left(\frac{s_i}{\delta_i^2}\right)
    =
    O\left(\frac{\widehat x^{2/3}s_i^{1/3}}{z^{4/3}}\right).
    \]
    Moreover, throughout phase $i$ we have $\delta_i\ge\delta_i^\star=s_i/\widehat x$, and hence
    \[
    p_i(\delta_i)
    \ge
    \min\left\{1,\frac{\lambda s_i}{\widehat x}\right\}.
    \]
    If $u$ could have been placed into $L_i$ by a late insertion, then
    $p_i(\delta_i)<1/2$ and $z\in[p_i(\delta_i),2p_i(\delta_i)]$, so
    $z\ge\lambda s_i/\widehat x$. For smaller $z$, the query does not need to
    inspect $L_i^{2-g(u)}$. Therefore,
    \[
    \mathbb E\left[d_i^{2-g(u)}\right]
    \le
    \int_{\min\{1,\lambda s_i/\widehat x\}}^1
    O\left(\frac{\widehat x^{2/3}s_i^{1/3}}{z^{4/3}}\right)\,dz
    =
    O(\widehat x).
    \]
    We conclude that $\mathbb E[d_i] = O(\widehat x)$, as desired.
\end{proof}

\begin{lemma}
  The expected cost of a negative query is
  \[O(\widehat x \log x) = O(x \log^2 x).\]
\end{lemma}

\begin{proof}
  A negative query finishes once it concludes for each $i$ that $u$ is not in $L_i$. Therefore, the expected cost of a negative query is
  \[\mathbb E\left[\max_{i \in I} d_i\right] \leq \sum_{i \in I} \mathbb E\left[d_i\right] = O(\widehat x |I|) = O(\widehat x \log x),\]
  as desired.
\end{proof}

        Putting the pieces together, we have the following theorem:

        \begin{theorem}
        There is an implementation of linear probing without reordering that has the following properties. If $x \le n^{0.49}$ is known in advance, then each of the first $(1 - 1/x)n$ insertions completes in worst-case expected time $O(x \log x)$, each positive query completes in worst-case expected time $O(x \log x)$, and each negative query completes in worst-case expected time $O(x \log^2 x)$.
            \label{thm:mainneg}
        \end{theorem}

\section{Optimizing Amortized Cost}
\label{sec:nonuniform-scales}

We now present the full version of interlinear probing, which achieves amortized expected insertion
cost $O(x)$. First, let us see why we do not already achieve amortized expected cost $O(x)$. Observe that each layer $L_i$ is filled to be $1-\delta_i^\star = 1-s_i/\widehat{x}$ full. By a standard linear probing analysis, the amortized expected number of probes made per insertion into $L_i$ is $\Theta(1/\delta_i^\star) = \Theta(\widehat x/s_i)$. Multiplying by $s_i$ (to account for the slots in $L_i$ being $s_i$ apart) gives an amortized expected cost of $\Theta(\widehat x)$ per insertion into $L_i$. Therefore, the amortized expected insertion cost across all layers is also $\Theta(\widehat x)$. 
Notice that the amortized cost of an insertion into each layer is roughly the same, and every layer ends up with the same number of free slots, $(n/s_i)\delta_i^\star = n / \widehat x$.

However, the amortized contribution to the overall amortized cost of insertions is \emph{not} the same for each layer: the total contribution to the overall amortized expected insertion cost from layers $L_i$ for $i \geq \log\log x$ is at most
\[
\frac{1}{n(1-1/x)} \sum_{i \geq \log\log x} O(\widehat x \cdot |L_i|) = O(x).
\]

Therefore, it is necessary to reduce the amortized cost of insertions into the first $\log\log x$ layers. We opt to replace these problematic layers with a single ``megalayer'' taking up a $1-1/\log x$ fraction of all slots, assuming $x$ is sufficiently large.\footnote{Otherwise, standard linear probing suffices for any constant $x$.} At the end of all insertions, this megalayer has $n/(2x)$ empty slots, and thus the amortized expected insertion cost is $O(x)$. The key to the construction is making sure that the worst-case expected insertion and negative query times remain the same after this change.



\paragraph{A two-layer decomposition.} We now partition the table into two layers $L_1$ and $L_2$, with
\[
|L_1|=n\left(1-\frac{1}{\log x}\right)
\qquad\text{and}\qquad
|L_2|=\frac{n}{\log x},
\]
where layer $L_2$ consists of one slot every $\log x$ slots and $L_1$ consists of the remaining slots. Consecutive slots of $L_1$ are at most two positions apart, so we may treat its spacing as $s_1=1$ at the cost of only constant factors.\footnote{Technically, projecting a uniform hash $h(u)$ onto the first slot of $L_1$ in its probe sequence would make a slot immediately after a slot of $L_2$ twice as likely to be the first probe into $L_1$ as any other slot. To restore uniformity, whenever an insertion routed to $L_1$ has $h(u)$ in $L_2$, we independently and uniformly select one of the $\log x-1$ slots of $L_1$ between $h(u)$ and the next slot of $L_2$. The insertion scans from $h(u)$ to the selected slot, skipping the intervening slots of $L_1$, and then performs ordinary linear probing restricted to $L_1$. This adds only $O(\log x)$ probes. For a negative query with $h(u)$ in $L_2$, we ignore the stopping conditions for $L_1$ until the query has scanned the first $\log x$ positions. Thus, this modification does not affect any of the worst-case or amortized bounds.} We ignore rounding throughout this section. Layer $L_2$ is further partitioned according to the algorithm we already specified for insertions, targeting a load factor of
\[
1-\frac{\log x}{2x}.
\]
Thus, there will be (at most)
\[
|L_2|\frac{\log x}{2x}=\frac{n}{2x}
\]
empty slots in $L_2$ at the end of all insertions. To invoke Theorem~\ref{thm:main} inside $L_2$, we use load parameter
\[
x_2:=\frac{2x}{\log x}.
\]
Under the assumption $x\le n^{0.49}$ used below, $x_2$ satisfies the hypotheses of Theorems~\ref{thm:main} and~\ref{thm:mainneg} for a table of size $|L_2|$. Multiplying the bound from Theorem~\ref{thm:main} by $\log x$ to account for the density of $L_2$ in the hash table, the worst-case expected cost of insertions into $L_2$ is
\[
O\left(\log x\cdot x_2\log x_2\right)=O(x\log x).
\]
Since at most $|L_2|=n/\log x$ elements are inserted into $L_2$, their contribution to the overall amortized insertion cost is
\[
O\left(\frac{x\log x\,|L_2|}{n(1-1/x)}\right)=O(x).
\]
We therefore focus on insertions into $L_1$ for the rest of the amortized analysis.

\paragraph{Two new phases.} We divide insertions into the hash table into two phases, 1 and 2, and set constants $\lambda$ and $C_0$ sufficiently large. Our insertion algorithm follows the same protocol as Algorithm~\ref{alg:non-greedy-insert}, restricted to the two layers $L_1,L_2$ and the two phases, with the difference that insertions into $L_2$ do not follow ordinary linear probing but instead invoke our original interleaved algorithm restricted to $L_2$. This internal algorithm is used from the first insertion routed to $L_2$, including during phase 1.

Set $s_1:=1$ and
\[
\delta_1^\star:=\frac{n}{2x|L_1|}=\Theta(1/x).
\]
Thus, $L_1$ also has $n/(2x)$ empty slots when phase 1 ends. Because we now jump directly from a layer of density $\Theta(1)$ to a layer of density $\Theta(1/\log x)$, we set
\[
\widehat x_{1}:=C_0x\log^2(2x).
\]
During phase 1, a key is routed to $L_1$ with probability
\[
p_1(\delta_1)
=
\min\left\{1,
\lambda\widehat x_{1}^{1/2}\delta_1^{3/2}
\right\},
\]
and is otherwise routed to $L_2$. Phase 1 ends when $\delta_1\le\delta_1^\star$; during phase 2, all remaining keys are routed to $L_2$. In effect, we set the probability of insertion as if we are going to fill it up to a factor $1-O(1/(x\log^2 x))$ full, but then we stop early, when $L_1$ is only a $1-\Theta(1/x)$ fraction full. As we shall see, increasing the probability of late and expensive insertions into $L_1$ balances out with the fact that $L_1$ fills up less than it would have under the original algorithm, keeping the worst-case expected costs the same as in Theorem \ref{thm:mainneg} while achieving better amortized bounds.

To support negative queries, we also apply the two modifications from Section~5. We split $L_1$ into two alternating sub-layers, $L_1^1$ and $L_1^2$, and use the same fully random hash functions $f$ and $g$ for both $L_1$ and each internal layer of $L_2$. As shown below, except in the critical-failure event, the first internal layer of $L_2$, which has size $|L_2|/2$, remains at most half full throughout phase 1, so these insertions are early insertions; thereafter, early and late insertions into each internal layer of $L_2$ are distinguished exactly as in Section~5.

\paragraph{Correctness and worst-case expected cost.} The two layers retain a total of
\[
\delta_1^\star |L_1|+|L_2|\frac{\log x}{2x}
=
\frac{n}{2x}+\frac{n}{2x}
=
\frac nx
\]
empty slots, so the algorithm supports the first $(1-1/x)n$ insertions.

It remains to check that $L_2$ does not fill prematurely during phase 1. Let $O_1$ be the number of keys routed to $L_2$ during this phase. Repeating the geometric-random-variable calculation from Claim~\ref{claim:overflow-count}, we have
\begin{align*}
\E[O_1]
&=
O\left(
\frac{|L_1|}{\lambda}
\sqrt{\frac{s_1}{\widehat x_{1}\delta_1^\star}}
\right)\\
&=
O\left(\frac{n}{\log(2x)}\right)
=O(|L_2|).
\end{align*}
By taking $C_0$ and $\lambda$ sufficiently large, the implicit constant can be made smaller than $1/8$. The same moment-generating-function argument as in Section~4 then gives
\[
\Prb\left[O_1>\frac{|L_2|}{4}\right]
\le
\exp\left(-\Omega(n/x)\right).
\]
Indeed, the minimum nontrivial routing probability is $\Theta(\log(2x)/x)$, and multiplying it by $|L_2|=\Theta(n/\log(2x))$ gives the exponent $\Omega(n/x)$. As in Section~4, this critical-failure event contributes negligibly when $x\le n^{0.49}$, since we can use the trivial $O(n)$ insertion bound on that event.

For an insertion made during phase 1, the expected contribution from the branch routed to $L_1$ is
\begin{align*}
p_1(\delta_1)O\left(\frac{s_1}{\delta_1^2}\right)
&=
O\left(\sqrt{\frac{\widehat x_{1}s_1}{\delta_1}}\right)\\
&\le
O\left(\sqrt{\widehat x_{1}x}\right)
=O(x\log x),
\end{align*}
where the inequality uses $\delta_1\ge\delta_1^\star=\Theta(1/x)$. As observed above, the expected physical cost of an insertion routed to the internal construction on $L_2$ is also $O(x\log x)$. Thus every insertion, and hence every positive query, has worst-case expected cost $O(x\log x)$.

The negative-query argument is the same as in Section~5. For the outer layer, we have throughout phase 1 that
\[
\delta_1\ge\delta_1^\star
=\Theta(1/x)
\ge \frac{s_1}{\widehat x_{1}}.
\]
For $L_1^{g(u)}$, early insertions stop when $p_1(\delta_1)$ falls below $1/2$, at which point
\[
\delta_1=\Theta\left(\widehat x_{1}^{-1/3}\right).
\]
The expected distance to a position certifying that $u\notin L_1^{g(u)}$ is therefore $O(\widehat x_{1}^{2/3})$. For $L_1^{2-g(u)}$, conditioning on $f(u)=z$, the same argument at the point where $p_1(\delta_1)$ falls below $z/4$ gives expected certification distance
\[
O\left(\frac{\widehat x_{1}^{2/3}}{z^{4/3}}\right).
\]
Moreover, a late insertion is possible only when $z\ge \lambda/\widehat x_{1}$. Averaging over $z$ therefore gives
\[
\int_{\lambda/\widehat x_{1}}^1
O\left(\frac{\widehat x_{1}^{2/3}}{z^{4/3}}\right)\,dz
=O(\widehat x_{1}).
\]
Thus the expected distance needed to certify that $u\notin L_1$ is $O(\widehat x_{1})=O(x\log^2 x)$. Applying Theorem~\ref{thm:mainneg} inside $L_2$ and multiplying distances by $\log x$ gives
\[
O\left(\log x\cdot x_2\log^2 x_2\right)
=O(x\log^2 x).
\]
Taking the maximum of these two certification distances, and upper-bounding its expectation by their sum, gives worst-case expected negative-query cost $O(x\log^2 x)$.

\paragraph{Amortized expected cost.} By the standard amortized analysis of linear probing, the expected total cost of filling $L_1$ down to empty fraction $\delta_1^\star$ is
\[
O\left(\frac{s_1|L_1|}{\delta_1^\star}\right)=O(nx).
\]
The balancing argument from Section~5 shows that splitting $L_1$ into two sub-layers changes this by only a constant factor. As calculated above, all insertions into $L_2$ also have total expected cost $O(nx)$. Thus, the expected total cost of the first $(1-1/x)n=\Theta(n)$ insertions is $O(nx)$, giving amortized expected insertion cost $O(x)$. Since elements never move after being inserted, the average expected positive-query cost is also $O(x)$.

Putting the pieces together completes the proof of Theorem~\ref{thm:main-amortized}.

\section{A Lower Bound}

In this section, we prove an $\Omega(x\sqrt{\log x})$ lower bound on the
worst-case expected insertion cost of any stable linear-probing insertion
strategy.  In fact, the lower bound holds in a more permissive model: elements
may be moved after insertion, and they may be placed on either side of their
home locations, provided that moving an element by distance $d$ incurs cost
$d$.  A stable insertion strategy is the special case in which previously
inserted elements are never moved.

\paragraph{The model.}
We consider a circular hash table with $n$ cells, identified with
$\mathbb Z_n$.  A sequence of elements
\[
    u_1,u_2,\ldots,u_N
\]
is inserted into the table, where $N=(1-\epsilon)n$.  We assume for
simplicity that $(1-\epsilon)n$ is an integer.  Each element $u_i$ has a
home location $h(u_i)\in\mathbb Z_n$, and the home locations are mutually
independent and uniformly distributed.  The value $h(u_i)$ is revealed to
the algorithm when $u_i$ is inserted; in particular, the algorithm does not
know the home locations of future elements.  For a randomized algorithm, we
may equivalently sample its entire random tape before the first insertion,
independently of the home locations.

During an insertion, the algorithm may move both the newly inserted element
and previously inserted elements.  We view the newly inserted element as
starting at its home location.  At the end of each insertion, every inserted
element must occupy a distinct table cell.  The \defn{cyclic distance}
between two cells is the smaller of their clockwise and counterclockwise
distances.  Moving an element between two cells costs their cyclic distance,
and the cost of an insertion is the sum of the costs of all moves performed
during that insertion.

For each $j\in\mathbb Z_n$, define \defn{boundary $j$} to be the boundary
between cells $j-1\bmod n$ and $j$.  For the analysis, represent every move by the shorter of the clockwise and counterclockwise routes, with ties broken arbitrarily.  A move of cost $d$ is thereby
represented by exactly $d$ boundary crossings.

\paragraph{The lower bound.}
We prove the following theorem.

\begin{theorem}\label{thm:online-placement-lower-bound}
There are absolute constants $C,\epsilon_0>0$ such that the following
holds.  Let
\[
    0<\epsilon<\epsilon_0
    \qquad\text{and}\qquad
    n\ge C\epsilon^{-2}\log(1/\epsilon).
\]
Consider any possibly randomized online algorithm that inserts
$N=(1-\epsilon)n$ elements into a circular table of size $n$ under the model
above.  The algorithm may know $n$, $\epsilon$, and $N$ in advance.  Then
there is an insertion $i\in\{1,\ldots,N\}$ whose expected cost is
\[
    \Omega\!\left(
        \frac{\sqrt{\log(1/\epsilon)}}{\epsilon}
    \right).
\]
The expectation is over the home locations and any internal randomness of
the algorithm.
\end{theorem}

Setting $\epsilon=1/x$ gives the claimed
$\Omega(x\sqrt{\log x})$ lower bound.

\subsection{Proof Overview}

The proof compares, at a carefully chosen scale, the imbalance generated by the random home locations with the smaller occupancy imbalance in the algorithm's final placement. The discrepancy between the two must be accounted for by movement across the blocks' boundaries.

Our terminology generally follows Proposition~19 of
\cite{bender2022online}.  At scale $s$, an $s$-block is a circular interval of exactly
$s$ cells.  For a starting cell $r$, we consider two adjacent $s$-blocks
$B_1(r,s)$ and $B_2(r,s)$.  We call the pair \emph{live} if its union has
density at least $1-100\epsilon$, and \emph{dead} otherwise.  The pair's
\emph{adjusted imbalance} is the absolute difference between the densities
of the two blocks when the pair is live, and zero when it is dead.  We write
$\Delta_s(A)$ for the average adjusted imbalance over the starting cell $r$
in an occupied set $A$.  Equivalently, our scale $s$ corresponds to the
scale $2s$ in the indexing convention of \cite{bender2022online}.

The proof has three steps.

\paragraph{1. Selecting a block size with small final imbalance.}
Section~\ref{sec:block-density-difference-sum} proves a deterministic
multiscale bound for the final occupied set.  A potential argument across a
nested family of power-of-two block sizes gives
\[
    \sum_s \Delta_s(A)^2=O(\epsilon^2).
\]
There are $\Theta(\log(1/\epsilon))$ relevant sizes, so there exists one
power-of-two value $s$ satisfying
\[
    \E[\Delta_s(A_{\mathrm{final}})]
    =O\!\left(
        \frac{\epsilon}{\sqrt{\log(1/\epsilon)}}
    \right).
\]
We fix this value of $s$ for the entire proof. It is important that $s$ is fixed using the expected final
imbalances $\E[\Delta_s(A_{\mathrm{final}})]$, rather than being selected from the realized final configuration. Thus, $s$ is a deterministic function of
the algorithm, $n$, and $\epsilon$, and is independent of the home
locations and random tape in any particular execution. Consequently,
after fixing $s$, conditioning on the home locations of earlier insertions does not reveal any information about the later insertions.

\paragraph{2. The final home locations create a larger imbalance.}
We then study the imbalance created during a length-$t$ suffix of the insertion sequence.
Section~\ref{sec:random-count-difference} shows that, regardless of the imbalance already present before the suffix, the random home locations in the suffix produce an uncorrected count imbalance between two adjacent $s$-blocks of order
\[
    \Omega\!\left(\sqrt{\frac{ts}{n}}\right)
\]
with constant probability, but the scale $s$ was chosen to have expected final count imbalance $O\!\left(
        {s\epsilon}/{\sqrt{\log(1/\epsilon)}}
    \right)$. Choosing
\[
    t
    \in\Theta\!\left(
        \frac{ns\epsilon^2}{\log(1/\epsilon)}
    \right)
\]
with a sufficiently large implicit constant makes the suffix imbalance exceed the expected final count difference by a constant factor with constant probability.

\paragraph{3. Correcting the imbalance requires movement.}
Section~\ref{sec:main-proof} combines the two ingredients.  For a uniformly
random pair of adjacent $s$-blocks, the event that the pair is dead has only
small constant probability.  On the complementary event, the large
imbalances created by the suffix home locations must be reconciled with the
small final occupancy difference.  This requires movement of elements across the blocks' boundaries by the algorithm to balance their occupancies.

Let $K$ be the total insertion cost during the final $t$ insertions.  The
expected number of relevant boundary crossings witnessed by a uniformly
random starting cell is $4K/n$.  It follows that
\[
    \E[K]=\Omega(\sqrt{nst}).
\]
Dividing by $t$ gives an average expected cost during the suffix of
\[
    \Omega\!\left(\sqrt{\frac{ns}{t}}\right)
    =
    \Omega\!\left(
        \frac{\sqrt{\log(1/\epsilon)}}{\epsilon}
    \right),
\]
and hence at least one insertion has the claimed expected cost.

\subsection{A Power-of-Two Scale with Small Final Occupancy Imbalance}
\label{sec:block-density-difference-sum}

We first analyze the final occupied set and select the spatial scale that
will later be used in the suffix argument.  No suffix length is fixed at this
stage.

For a set $A\subseteq\mathbb Z_n$ of occupied cells and a circular interval
$B$, define its density by
\[
    \rho_A(B)=\frac{|A\cap B|}{|B|}.
\]
An \defn{$s$-block} is a circular interval containing exactly $s$ cells.  For
$r\in\mathbb Z_n$ and an integer $s$ with $2s\le n$, let
\[
\begin{aligned}
    B_1(r,s)&=\{r,r+1,\ldots,r+s-1\},\\
    B_2(r,s)&=\{r+s,r+s+1,\ldots,r+2s-1\},
\end{aligned}
\]
where all indices are taken modulo $n$.  Thus $B_1(r,s)$ and $B_2(r,s)$ are
two adjacent $s$-blocks.  Write
\[
    U(r,s)=B_1(r,s)\cup B_2(r,s)
\]
for their union, which contains $2s$ cells.

We call the pair $(B_1(r,s),B_2(r,s))$ \defn{live} if
$\rho_A(U(r,s))\ge 1-100\epsilon$, and \defn{dead} otherwise.  Define the
pair's adjusted imbalance by
\[
    \Delta(r,s;A)
    =
    \left|
        \rho_A(B_1(r,s))-\rho_A(B_2(r,s))
    \right|
    \mathbf 1_{\{\rho_A(U(r,s))\ge 1-100\epsilon\}},
\]
and define the average adjusted imbalance at scale $s$ by
\[
    \Delta_s(A)
    =
    \frac1n\sum_{r\in\mathbb Z_n}\Delta(r,s;A).
\]
In other words, $\Delta_s(A)$ averages the density difference between two
adjacent $s$-blocks, while dead pairs contribute zero.

\begin{lemma}\label{lem:block-density-difference-sum}
For every $0<\epsilon\le 1/100$ and every
$A\subseteq\mathbb Z_n$ with $|A|=(1-\epsilon)n$,
\[
    \sum_{\substack{s\text{ a power of two}\\s\le n/8}}
        \Delta_s(A)^2
    =O(\epsilon^2).
\]
\end{lemma}

The lemma is deterministic and holds for every fixed occupied set $A$.
It is a direct result of Proposition~19 of \cite{bender2022online}; we
include the proof for completeness.

\begin{proof}
Set $q=1-100\epsilon$.  For every circular interval $I$, define the potential
\[
    \Phi(I)=\bigl(\rho_A(I)-q\bigr)_+^{\;2},
    \qquad
    (z)_+=\max\{z,0\}.
\]
Thus the potential is zero below the density threshold $q$ and is the squared
excess above $q$ otherwise.

Let $I_1,I_2$ be two adjacent blocks of equal length, let
$I=I_1\cup I_2$, and set
\[
    \delta=\left|\rho_A(I_1)-\rho_A(I_2)\right|.
\]
After possibly exchanging $I_1$ and $I_2$, we have
\[
    \rho_A(I_1)=\rho_A(I)+\frac{\delta}{2},
    \qquad
    \rho_A(I_2)=\rho_A(I)-\frac{\delta}{2}.
\]
We claim that
\begin{equation}\label{eq:potential-increment}
    \frac{\Phi(I_1)+\Phi(I_2)}{2}
    \ge
    \Phi(I)
    +
    \frac{\delta^2}{8}
    \mathbf 1_{\{\rho_A(I)\ge q\}}.
\end{equation}

Suppose first that $\rho_A(I)\ge q$, and set
$a=\rho_A(I)-q\ge0$.  Then
\[
    \frac{\Phi(I_1)+\Phi(I_2)}{2}
    =
    \frac12\left(a+\frac{\delta}{2}\right)_+^2
    +
    \frac12\left(a-\frac{\delta}{2}\right)_+^2.
\]
If $a\ge\delta/2$, this expression is
$a^2+\delta^2/4$.  If $a<\delta/2$, it is
$\frac12(a+\delta/2)^2$, and
\[
    \frac12\left(a+\frac{\delta}{2}\right)^2
    -a^2-\frac{\delta^2}{8}
    =
    \frac{a(\delta-a)}{2}
    \ge0.
\]
Thus \eqref{eq:potential-increment} holds when $\rho_A(I)\ge q$.  If
$\rho_A(I)<q$, then $\Phi(I)=0$ and the indicator vanishes, while
$\Phi(I_1),\Phi(I_2)\ge0$, so the inequality again holds.

If \(n<8\), the sum is empty and the result is immediate. Hence, assume \(n\ge8\), and let $s_{\max}$ be the largest power of two satisfying
$s_{\max}\le n/8$. Choose $I_0$ uniformly among all circular intervals of
length $2s_{\max}$.  Given $I_j$, partition it into two adjacent intervals
of equal length and choose $I_{j+1}$ uniformly from these two intervals.
Continue until reaching a single-cell interval $I_{\mathrm{final}}$.  At
every level, $I_j$ is uniformly distributed among all circular intervals of
its length.

At the level where $I_j$ has length $2s$,
\eqref{eq:potential-increment} gives
\[
    \E\left[\Phi(I_{j+1})-\Phi(I_j)\right]
    \ge
    \frac18
    \E\left[
        \delta^2\mathbf 1_{\{\rho_A(I_j)\ge q\}}
    \right].
\]
By Cauchy--Schwarz,
\[
\begin{aligned}
    \E\left[
        \delta^2\mathbf 1_{\{\rho_A(I_j)\ge q\}}
    \right]
    &\ge
    \left(
        \E\left[
            \delta\mathbf 1_{\{\rho_A(I_j)\ge q\}}
        \right]
    \right)^2\\
    &=\Delta_s(A)^2.
\end{aligned}
\]
Summing over the levels therefore yields
\[
    \E[\Phi(I_{\mathrm{final}})]-\E[\Phi(I_0)]
    \ge
    \frac18
    \sum_{\substack{s\text{ a power of two}\\s\le n/8}}
        \Delta_s(A)^2.
\]
Since $\Phi(I_0)\ge0$ and $I_{\mathrm{final}}$ is a single cell,
$\Phi(I_{\mathrm{final}})\le(100\epsilon)^2$.  Hence
\[
    \sum_{\substack{s\text{ a power of two}\\s\le n/8}}
        \Delta_s(A)^2
    \le
    8\cdot100^2\epsilon^2,
\]
which proves the lemma.
\end{proof}

\subsection{Random Home Locations Create Occupancy Imbalance}
\label{sec:random-count-difference}

We next isolate the probabilistic fact used after the block size has been
selected.  The statement applies to arbitrary $s$ and $t$, which will later be fixed in
Section~\ref{sec:main-proof}.

Consider two disjoint blocks $B_1,B_2$, each of size $s$, and the final $t$
insertions.  Before these insertions, let $d$ be the difference between the
numbers of occupied cells in $B_1$ and $B_2$.  Let $Z$ be the number of the
final $t$ home locations in $B_1$ minus the number in $B_2$.  The next lemma
shows that, regardless of the fixed value $d$, the magnitude of $d+Z$ is at
least a constant multiple of the standard deviation of $Z$ with constant
probability.

\begin{lemma}\label{lem:random-count-difference}
Let $n,s,t$ be positive integers with $2s\le n$, and let
$Z=X_1+\cdots+X_t$, where $X_1,\ldots,X_t$ are independent random variables
satisfying
\[
    X_i=
    \begin{cases}
        1  & \text{with probability }s/n,\\
        -1 & \text{with probability }s/n,\\
        0  & \text{with probability }1-2s/n.
    \end{cases}
\]
Set
\[
    \sigma^2=\E[Z^2]=\frac{2ts}{n}.
\]
If $\sigma^2\ge1$, then, for every $d\in\mathbb R$,
\[
    \Prb\!\left[
        |d+Z|\ge\frac{\sigma}{\sqrt2}
    \right]
    \ge\frac1{16}.
\]
\end{lemma}

This is a special case of the general anti-concentration principle that a
sum of independent bounded random variables with standard deviation $\tilde\sigma >1$ has constant probability of being a constant fraction of $\tilde\sigma$ away from any fixed shift; see, for example, Theorem~12 of
\cite{kuszmaul2025simple}.  We include a direct proof.

\begin{proof}
The distribution of $Z$ is symmetric.  Moreover,
\[
    \E[X_i^2]=\E[X_i^4]=\frac{2s}{n}.
\]
By independence and symmetry,
\[
\begin{aligned}
    \E[Z^4]
    &=
    \sum_{i=1}^t\E[X_i^4]
    +
    6\sum_{1\le i<j\le t}\E[X_i^2]\E[X_j^2]\\
    &\le
    \sigma^2+3\sigma^4
    \le
    4\sigma^4,
\end{aligned}
\]
where the final inequality uses $\sigma^2\ge1$.

Fix $d\in\mathbb R$ and set $Y=d+Z$.  Then
\[
    \E[Y^2]=d^2+\sigma^2
\]
and
\[
\begin{aligned}
    \E[Y^4]
    &=d^4+6d^2\sigma^2+\E[Z^4]\\
    &\le d^4+6d^2\sigma^2+4\sigma^4\\
    &\le 4(d^2+\sigma^2)^2
    =4\E[Y^2]^2.
\end{aligned}
\]
Let
\[
    \mathcal E
    =
    \left\{
        Y^2\ge\frac12\E[Y^2]
    \right\}.
\]
Since $Y^2<\frac12\E[Y^2]$ outside $\mathcal E$,
\[
    \frac12\E[Y^2]
    \le
    \E[Y^2\mathbf 1_{\mathcal E}]
    \le
    \sqrt{\E[Y^4]\Prb[\mathcal E]}.
\]
Consequently,
\[
    \Prb[\mathcal E]
    \ge
    \frac{\E[Y^2]^2}{4\E[Y^4]}
    \ge
    \frac1{16}.
\]
On $\mathcal E$,
\[
    |Y|
    \ge
    \sqrt{\frac{d^2+\sigma^2}{2}}
    \ge
    \frac{\sigma}{\sqrt2},
\]
which proves the lemma.
\end{proof}

\subsection{Correcting the Final-Suffix Imbalance Requires Movement}
\label{sec:main-proof}

We now combine the two ingredients:  First, we fix a scale $s$ with small final imbalance using the multiscale bound. We then choose the suffix length $t$ as a function of $s$, so that the random imbalance contributed by the suffix home locations has the required magnitude at scale $s$. Finally, we average over a uniformly random starting position $r$ to finish the proof.

\begin{proof}[Proof of Theorem~\ref{thm:online-placement-lower-bound}]
Let
\[
    \ell=\log(1/\epsilon),
\]
and let $b\ge1$ be a constant to be chosen below.

\paragraph{Choosing a block size with small final imbalance.}
Let
\begin{equation}\label{eq:block-size-range}
    \mathcal S
    =
    \left\{
        s\text{ a power of two}:
        \frac{\sqrt\ell}{\epsilon}
        \le s\le
        \frac{\ell}{100b\epsilon^2}
    \right\}.
\end{equation}
The ratio of the upper and lower endpoints is
$\sqrt\ell/(100b\epsilon)$.  Thus, after choosing $\epsilon_0>0$
sufficiently small, $|\mathcal S|=\Theta(\ell)$.  After choosing $C$
sufficiently large, every $s\in\mathcal S$ also satisfies $s\le n/8$.

Let $A_{\mathrm{final}}\subseteq\mathbb Z_n$ be the random occupied set after
all insertions, and define
\[
    \overline\Delta_s
    =
    \E[\Delta_s(A_{\mathrm{final}})],
\]
where the expectation is over the home locations and the algorithm's randomness.  By Lemma~\ref{lem:block-density-difference-sum} and Jensen's
inequality,
\[
\begin{aligned}
    \sum_{s\in\mathcal S}\overline\Delta_s^2
    &\le
    \sum_{s\in\mathcal S}
        \E[\Delta_s(A_{\mathrm{final}})^2]\\
    &=O(\epsilon^2).
\end{aligned}
\]
Since $|\mathcal S|=\Theta(\ell)$, some fixed $s\in\mathcal S$ satisfies
\begin{equation}\label{eq:small-final-density-difference}
    \overline\Delta_s
    =O\!\left(\frac{\epsilon}{\sqrt\ell}\right).
\end{equation}

\paragraph{Choosing a suffix whose random imbalance dominates.}
Fix this value of $s$, and set
\begin{equation}\label{eq:number-of-final-insertions}
    t
    =
    \left\lfloor
        \frac{bns\epsilon^2}{\ell}
    \right\rfloor.
\end{equation}
By the upper bound on $s$,
\[
    t
    \le
    \frac{bns\epsilon^2}{\ell}
    \le
    \frac n{100}.
\]
Choosing $\epsilon_0\le1/2$ ensures that
$N=(1-\epsilon)n\ge n/2$, so the final suffix of length $t$ is well-defined.

The assumptions $n\ge C\epsilon^{-2}\ell$ and
$s\ge\sqrt\ell/\epsilon$ imply
\[
    \frac{bns\epsilon^2}{\ell}
    \ge
    \frac{bC\sqrt\ell}{\epsilon}.
\]
Choosing $\epsilon_0$ sufficiently small makes the right-hand side at least
$2$, and hence
\begin{equation}\label{eq:t-lower-bound}
    t
    \ge
    \frac{bns\epsilon^2}{2\ell}.
\end{equation}

\paragraph{Unseen home locations create a large occupancy difference.}
Choose $r\in\mathbb Z_n$ uniformly at random, independently of the home
locations and the algorithm's random tape, and let
\[
    B_1=B_1(r,s),
    \qquad
    B_2=B_2(r,s),
    \qquad
    U=U(r,s)=B_1\cup B_2.
\]
Let $D_0$ be the number of occupied cells in $B_1$ minus the number in $B_2$
immediately before the final $t$ insertions, and let $D_1$ be the corresponding
difference after all insertions.  Let $Z$ be the number of the final $t$ home
locations in $B_1$ minus the number in $B_2$.

Condition on $r$, the home locations of the first $N-t$ elements, and the
algorithm's random tape.  Because the algorithm is online, this conditioning
fixes $D_0$, while the final $t$ home locations remain independent and uniform.
Thus $Z$ has the distribution in Lemma~\ref{lem:random-count-difference} with
\begin{equation}\label{eq:sigma-lower-bound}
\sigma^2
    :=
    \E[Z^2]
    =
    \frac{2ts}{n}, \qquad
    \sigma
    = \sqrt{\frac{2ts}{n}} \ge
    \sqrt b\,\frac{s\epsilon}{\sqrt\ell},
\end{equation}
where the last inequality is from \eqref{eq:t-lower-bound}.
Since $s\ge\sqrt\ell/\epsilon$, this also gives
$\sigma^2\ge b\ge1$.

Let $\mathcal L$ be the event that the pair $(B_1,B_2)$ is live in the final
configuration:
\[
    \mathcal L
    =
    \left\{
        \rho_{A_{\mathrm{final}}}(U)
        \ge
        1-100\epsilon
    \right\}.
\]
For any fixed final configuration, each empty cell belongs to $U(r,s)$ for
exactly $2s$ of the $n$ choices of $r$.  Hence the expected number of empty
cells in $U$, over the uniform choice of $r$, is $2s\epsilon$.  If
$\mathcal L$ fails, $U$ contains more than $200s\epsilon$ empty cells.
Markov's inequality therefore gives
\begin{equation}\label{eq:low-density-probability}
    \Prb[\mathcal L^c]
    \le
    \frac1{100}.
\end{equation}

Lemma~\ref{lem:random-count-difference} applies under every value of the
conditioning above, and hence also unconditionally:
\[
    \Prb\!\left[
        |D_0+Z|\ge\frac{\sigma}{\sqrt2}
    \right]
    \ge
    \frac1{16}.
\]
By the union bound,
\[
    \Prb\!\left[
        |D_0+Z|\ge\frac{\sigma}{\sqrt2}
        \ \text{and}\ 
        \mathcal L
    \right]
    \ge
    \frac1{16}-\frac1{100}.
\]
Consequently, for an absolute constant $c_1>0$,
\begin{equation}\label{eq:initial-plus-arrivals}
    \E\left[
        |D_0+Z|\mathbf 1_{\mathcal L}
    \right]
    \ge
    c_1\sigma.
\end{equation}

\paragraph{Restoring the final balance requires boundary crossings.}
For each boundary $j\in\mathbb Z_n$, let $T_j$ be the number of times it is crossed during the final $t$ insertions, in either direction. Define the crossings witnessed at $U(r,s)$ by
\[
    W_r=T_r+2T_{r+s}+T_{r+2s},
\]
where the shared boundary between $B_1$ and $B_2$ is counted twice for each crossing over it. As the net change in the difference between the counts of $B_1$ and $B_2$ caused by boundary crossings is $D_1-(D_0+Z)$,
\begin{equation}\label{eq:count-difference-identity}
W_r
\ge
\left|D_1-(D_0+Z)\right|
\ge
|D_0+Z|-|D_1|.
\end{equation}
Since $W_r\ge0$, restricting to the event $\mathcal L$ and taking expectations gives
\begin{equation}\label{eq:movement-lower-bound-1}
\E[W_r]
\ge
\E\left[
|D_0+Z|\mathbf 1_{\mathcal L}
\right]
-
\E\left[
|D_1|\mathbf 1_{\mathcal L}
\right].
\end{equation}

On $\mathcal L$, the final count difference is $s$ times the final density difference. Hence
\[
    \E\left[
        |D_1|\mathbf 1_{\mathcal L}
    \right]
    =
    s\overline\Delta_s
    =
    O\!\left(\frac{s\epsilon}{\sqrt\ell}\right)
    =
    O\!\left(\frac{\sigma}{\sqrt b}\right),
\]
where the last step uses \eqref{eq:sigma-lower-bound}. By choosing $b$ sufficiently large, this term is at most half the lower bound in \eqref{eq:initial-plus-arrivals}. Substituting the two bounds into \eqref{eq:movement-lower-bound-1} yields
\begin{equation}\label{eq:movement-lower-bound-2}
\E[W_r]
=
\Omega(\sigma).
\end{equation}

\paragraph{Averaging boundary crossings into movement cost.}
Let $K$ be the total movement cost during the final $t$ insertions.  For every
fixed execution,
\[
    K=\sum_{j\in\mathbb Z_n}T_j.
\]
Averaging $W_r$ over the uniform choice of $r$ gives, pointwise,
\[
\begin{aligned}
    \E_r[W_r]
    &=
    \frac1n\sum_{r\in\mathbb Z_n}
        \bigl(T_r+2T_{r+s}+T_{r+2s}\bigr)\\
    &=\frac{4K}{n}.
\end{aligned}
\]
Taking expectation over the execution and using
\eqref{eq:movement-lower-bound-2},
\[
    \frac4n\E[K]
    =
    \E[W_r]
    =
    \Omega(\sigma).
\]
Consequently,
\[
    \E[K]
    =
    \Omega(n\sigma)
    =
    \Omega(\sqrt{nst}).
\]

If $\mathcal C_i$ denotes the cost of insertion $i$, then
\[
    \frac{\E[K]}{t}
    =
    \frac1t
    \sum_{i=N-t+1}^{N}\E[\mathcal C_i].
\]
Using $t\le bns\epsilon^2/\ell$, we conclude that
\[
    \frac{\E[K]}{t}
    =
    \Omega\!\left(\sqrt{\frac{ns}{t}}\right)
    =
    \Omega\!\left(\frac{\sqrt\ell}{\epsilon}\right)
    =
    \Omega\!\left(
        \frac{\sqrt{\log(1/\epsilon)}}{\epsilon}
    \right).
\]
Thus at least one of the final $t$ insertions has the claimed expected cost.

The online assumption is used precisely when we condition on the first
$N-t$ home locations and the algorithm's random tape: under this conditioning,
$D_0$ is fixed, while the final $t$ home locations remain independent and
uniformly distributed.
\end{proof}

\ifanonymous
\else
\section*{Acknowledgments}
This work was partially supported by the Jeanne~B.\ and Richard~F.\ Berdik ARCS Pittsburgh Endowed Scholar Award, NSF grants CNS-2504471 and CCF-254216, and a Jane Street Research Grant.
\fi

\clearpage
\bibliographystyle{alpha}
\bibliography{main.bib}

\end{document}